\documentclass[11pt]{article}

\usepackage[latin1]{inputenc}
\usepackage[T1]{fontenc}
\usepackage[english]{babel}
\usepackage{amsfonts}
\usepackage{amssymb}
\usepackage{amsmath}
\usepackage{amsthm}
\usepackage{graphicx}

\newtheorem{proposition}{Proposition}
\newtheorem{lemma}{Lemma}
\newtheorem{theorem}{Theorem}

\newtheorem{remark}{Remark}
\newcommand{\meanv}[1]{\left\langle#1\right\rangle}

\def\be{\begin{equation}}
\def\ee{\end{equation}}
\def\bc{\begin{center}}
\def\ec{\end{center}}

\newcommand{\OO}[1]{O \left(\frac{1}{#1}\right)}

\begin{document}

\title{Equilibrium statistical mechanics of bipartite spin systems}


\author{Adriano Barra\footnote{Dipartimento di Fisica, Sapienza Universit\`a di Roma \& Gruppo Nazionale Fisica Matematica, Sezione Roma1},
Giuseppe Genovese \footnote{Dipartimento di Fisica \& Dipartimento di Matematica,
    Sapienza Universit\`a di Roma} \ and Francesco Guerra \footnote{Dipartimento di Fisica, Sapienza Universit\`a di Roma \& Istituto Nazionale Fisica Nucleare, Sezione Roma1}}

\date{October 2010}

\maketitle

\begin{abstract}
Aim of this paper is to give an extensive treatment of bipartite
mean field spin systems, ordered and disordered: at first, bipartite
ferromagnets are investigated, achieving an explicit expression
for the free energy
trough a new minimax variational principle. Furthermore via the Hamilton-Jacobi technique the same free energy structure is obtained together with the existence of its thermodynamic limit and the minimax principle is connected to a standard max one.
\newline
The same is investigated for bipartite spin-glasses: By the
Borel-Cantelli lemma a control of the high temperature regime is
obtained, while via the double stochastic stability technique we
get also the explicit expression of the free energy at the replica
symmetric level, uniquely defined by a minimax variational
principle again.
\newline
A general results that states that the free energies of these
systems are convex linear combinations of their independent one
party model counterparts is achieved too.
\newline
For the sake of completeness we show further that at zero
temperature the replica  symmetric entropy becomes negative and,
consequently,  such a symmetry must be broken. The treatment of the fully broken replica symmetry case is deferred to a forthcoming paper. As a
first step in this direction,   we start
deriving the linear and quadratic constraints to overlap
fluctuations.
\begin{center}
\end{center}

%\keywords{Bipartite Models \and Interpolation Techniques \and
%Continuous Spin}
\end{abstract}

\section{Introduction}

The investigation by statistical mechanics of simple and complex
mean field spin systems is experiencing  a huge increasing
interest in the last decades. The motivations are at least
two-fold: from one side, at the rigorous mathematical level, even
though several contributions appeared along the years (see for
instance
\cite{broken}\cite{limterm}\cite{talagrand}\cite{aizenman}), a
full clear picture is still to be achieved (it is enough to think
at the whole community dealing with ultrametricity in the case of
random interactions as in glasses
\cite{MPV}\cite{panchenko}\cite{arguin}); at the applied level,
these toy models are starting to be used in several different
contexts, ranging from quantitative sociology
\cite{BC1}\cite{brock}\cite{sw} to theoretical biology
\cite{parisi}\cite{barragliari}.
\newline
It is then obvious the need for always stronger and simpler
methods to analyze the enormous amount of all the possible
''variations on theme'', the theme being the standard Curie-Weiss
model (CW) \cite{barra0}\cite{leshouches} for the simpler systems,
or the paradigmatic Sherrington-Kirkpatrick model (SK)
\cite{MPV}\cite{broken} for the complex ones.
\newline
As a result, inspired by recent attention payed on two groups in
interactions (i.e. decision making processes in econometrics
\cite{macfadden}\cite{bip} or metabolic networks in biology
\cite{enzo}\cite{pagnani}), we decided to focus, in this paper, on
the equilibrium statistical mechanics of two bipartite spin
systems, namely the bipartite CW and the bipartite SK.
\newline
At first we approach to the problem of bipartite model by studying
in Section 2 the bipartite ferromagnet, obtaining, both via a
standard approach and through our mechanical interpretation of the
interpolation method
\cite{sum-rules}\cite{barra0}\cite{io1}\cite{dibiasio}, a
variational principle for the free energy in thermodynamic limit.
\newline
In Section 3 we open the investigation of the bipartite spin glass
model, and following the path yet outlined in
\cite{BG1}\cite{BGG1} we get the annealed free energy (more
precisely, the pressure), with a characterization of the region of
the phase diagram where it coincides with the true one in the
thermodynamic limit, and the replica symmetric free energy, by the
double stochastic stability technique, which stems from a minimax
variational principle, whose properties are discussed too.
\newline
Finally we calculate the zero temperature observable and, by
noticing that the entropy is negative defined, we conclude that
replica symmetry must be broken. Despite a full replica symmetry
breaking scheme deserves a whole work, here we start introducing
its typical linear and quadratic constraints, obtained by Landau
self-averaging of the internal energy \cite{GG}.
\newline
A last section is left for conclusions and outlooks.

\section{Bipartite ferromagnets}

We are interested in considering a set of $N$ Ising spin
variables, in which it is precisely defined a partition in two
subsets of size respectively $N_1$ and $N_2$. We assume the
variable's label of the first subset as $\sigma_i$, $i=1, ...,
N_1$, while the spins of the second one are introduced by
$\tau_j$, $j=1, ..., N_2$.  Of course we have $N_1+N_2=N$, and we
name the relative size of the two subsets $N_1/N=\alpha_N$,
$N_2/N=1-\alpha_N$.
\newline
For the sake of simplicity, in what follows we deal with both
parties formed by dichotomic variables, while we stress that the
method works on a very general class of random variables with
symmetric probability measure and compact support  \cite{io1}.
\newline
The spins interact via the Hamiltonian $H_{N_1,N_2}(\sigma, \tau, h_1, h_2)$, with $h_1\geq 0, \ h_2 \geq 0$:
$$
H_{N_1, N_2}(\sigma, \tau, h_1,
h_2)=-\frac{2}{N_1+N_2}\sum_{i=1}^{N_1}\sum_{j=1}^{N_2}\sigma_i\tau_j
- h_1\sum_{i=1}^{N_1}\sigma_i-h_2\sum_{j=1}^{N_2}\tau_j.
$$
We notice that spins in each subsystem interact only with spins in
the other one, but not among themselves; we have chosen to skip
the self interactions to tackle only genuine features stemming
from the exchange ones. The reader interested in a (different)
treatment of bipartite ferromagnetic models  with self
interactions may refer to \cite{bip}.
\newline
Partition function $Z$, pressure $A$ and free energy per site $f$ are defined naturally
for the model:
\begin{eqnarray}
Z_{N_1, N_2}(\beta, h_1, h_2)&=& \sum_{\sigma}\sum_{\tau}e^{-\beta H_{N_1, N_2}(\sigma, \tau, h_1, h_2)},\nonumber\\
A_{N_1, N_2}(\beta, h_1, h_2)&=&\frac{1}{N_1+N_2}\log Z_{N_1, N_2}(\beta, h_1, h_2),\nonumber\\
f_{N_1, N_2}(\beta, h_1, h_2)&=&-\frac{1}{\beta}A_{N_1, N_2}(\beta, h_1, h_2),\nonumber
\end{eqnarray}
while the thermodynamic limit of $A,f$ will be denoted via
$A(\alpha,\beta,h_1h_2)=\lim_N A_{N_1,N_2}(\beta,h_1,h_2)$ and
$f(\alpha,\beta,h_1,h_2)= \lim_N f_{N_1,N_2}(\beta,h_1,h_2)$,
where we stressed the prescription adopted in taking the infinite
volume limit, performed in such a way that when
$N,N_1,N_2\to\infty$, $N_1/N\to\alpha\in(0,1)$, and $N_2/N\to
1-\alpha \in (0,1)$. Taken $z(\sigma,\tau)$ as a generic function
of the spin variables, we can also specify the Boltzmann state of
our system as
\begin{equation}
\meanv{z(\sigma, \tau)}=\frac{\sum_{\sigma}\sum_{\tau} z(\sigma,
\tau)\exp(-\beta H_{N_1, N_2}(\sigma, \tau, h_1, h_2))}{Z_{N_1, N_2}(\beta, h_1,
h_2)}\label{eq:state-bip}.
\end{equation}
As usual the order parameters (the respective magnetizations of
the two systems) are
\begin{eqnarray}
m_{N_1}=\frac{1}{N_1}\sum_i^{N_1}\sigma_i,\\ \label{eq:m-bip}
n_{N_2}=\frac{1}{N_2}\sum_j^{N_2}\tau_j, \label{eq:n-bip}
\end{eqnarray}
thus the Hamiltonan reads off as \be\label{hamilto1} H_{N_1,
N_2}(\sigma,\tau, h_1, h_2)=-N\left[2\frac{\alpha_N}{1+\alpha_N}
m_{N_1}n_{N_2}+h_1m_{N_1}+h_2\alpha_Nn_{N_2}\right]. \ee
\begin{remark}
As will be clear soon, the choice of the factor $2$ in the
Hamiltonian is made in such a way that the balanced bipartite
model with $\alpha=1/2$ has the same critical point of the one
party model, i.e. $\beta=1$.
\end{remark}

\subsection{The occurrence of a minimax principle for the free energy}

Now we give the explicit form of the pressure of the model,
together with some interesting properties. The main result is the
following:
\begin{theorem}
In the thermodynamic limit, the pressure of the bipartite ferromagnetic model is given by the following variational principle
\begin{equation}
A(\alpha, \beta,h_1,h_2)=\max_{\bar{m}}\min_{\bar{n}} A_{trial}(\bar{m},\bar{n}),
\end{equation}
where
\begin{equation} A_{trial}(\bar{m},\bar{n})=
\log2 +\alpha\log\cosh\left(2\beta(1-\alpha) \bar{n} + h_1\right)
+(1-\alpha)\log\cosh\left( 2\beta\alpha \bar{m} + h_2
\right)-2\beta\alpha(1-\alpha)\bar{m}\bar{n}.\nonumber
\end{equation}
Furthermore, the solution is uniquely defined by the intersection
of
\begin{eqnarray}
\bar m&=&\tanh\left( 2\beta(1-\alpha)\bar n + h_1 \right)\label{eq:self1},\\
\bar n&=&\tanh\left( 2\beta\alpha\bar m + h_2
\right)\label{eq:self2},
\end{eqnarray}
for $\bar{m}\geq 0$ and $\bar{n} \geq 0$.
\end{theorem}
\begin{proof}
The proof can be achieved in several ways (i.e. a direct approach
by marginalizing the free energy with respect to both the
parties); we chose to follow the path outlined in \cite{BGG1} as
this may act as a guide later, dealing with frustrated
interactions.
\newline
For the sake of convenience, we rename $\beta h_1 \to h_1$ and
$\beta h_2 \to h_2$ as switching back to the original variables is
straightforward in every moment but this lightens the notation.
\newline
To our task we need to introduce two trial parameters
$\bar{m},\bar{n}$ mimic the magnetizations inside each party, one
interpolating parameter $t \in[0,1]$ and an interpolating function
$A(t)$ for the free energy as follows
\begin{equation}
A(t) = \frac{1}{N}\log \sum_{\sigma}\sum_{\tau}e^{t\Big( 2\beta
\frac{N_1 N_2}{N_1+N_2}m(\sigma)n(\tau) \Big)}e^{(1-t)\Big(
2\beta(1-\alpha)\bar{n}\sum_i^{N_1} \sigma_i + 2 \beta \alpha
\bar{m} \sum_j^{N_2} \tau_j \Big)}e^{h_1\sum_i \sigma_i + h_2
\sum_j \tau_j},
\end{equation}
such that, for $t=1$ our interpolating function reduces to the
free energy of the model (the pressure strictly speaking), while
for $t=0$ reduces to a sum of one-body models whose solution is
straightforward.
\newline
We can then apply the fundamental theorem of calculus to get the
following sum rule
\begin{equation}
A(1)= A(0) + \int_0^1 \frac{dA(t)}{dt} dt.
\end{equation}
To quantify the latter we need to sort out both the streaming of
the interpolating function as well as its value at $t=0$, namely
\begin{eqnarray}\label{streaming}
\partial_t A(t) &=& 2\beta \alpha (1-\alpha)[\langle m n \rangle - \bar{n} \langle m\rangle- \bar{m} \langle n\rangle + \bar{m}\bar{n}
] - \beta \alpha(1-\alpha)\bar{m}\bar{n},\\
A(t=0)&=&\log 2 + \alpha
\log\cosh[2\beta(1-\alpha)\bar{n}+h_1]+)(1-\alpha)\log\cosh[2\beta
\alpha \bar{m}+h_2],
\end{eqnarray}
where in equation (\ref{streaming}) we added and subtracted the
term $\beta \alpha(1-\alpha)\bar{m}\bar{n}$ so to write explicitly
the sum rule in terms of a trial function
$A_{trial}(\alpha,\beta,h_1,h_2)$ and an error term
$S(\bar{m},\bar{n})$:
\begin{equation}
A(\alpha,\beta,h_1,h_2)=
A_{trial}(\alpha,\beta,h_1,h_2)+S(\bar{m},\bar{n}),
\end{equation}
where
\begin{eqnarray}\nonumber
A_{trial}(\alpha,\beta,h_1,h_2) &=& \log 2 + \alpha \log \cosh
[2\beta (1-\alpha) \bar{n} + h_1 ] + (1-\alpha)\log \cosh [2\beta
\alpha \bar{m} + h_2] \\ &-& 2\beta
\alpha(1-\alpha)\bar{m}\bar{n}, \\
S(\bar{m},\bar{n}) &=& \beta \alpha (1-\alpha) \int_0^1 \langle (m
-\bar{m})(n-\bar{n}) \rangle_t dt,
\end{eqnarray}
for every trial functions $\bar{m},\bar{n}$.
\newline
Note that the averages $\langle . \rangle_t$ take into account
that the Boltzmannfaktor is no longer the standard one introduced
in eq.$(1)$, but incorporates the interpolating structure tuned by
the parameter $t$.
\newline
We stress that, at this stage, the error term $S(\bar{m},\bar{n})$
is the source of the fluctuations of the order parameters (which
are expected to reduce to zero in the thermodynamic limit and give
the label $S$) is not trivially positive defined as in many other
cases, even hardly to investigate (i.e. mono-party spin-glasses
\cite{broken}), however the idea of choosing properly
$\bar{m},\bar{n}$ so to make it smaller and smaller (eventually
zero) still holds obviously.
\newline
So we study at fixed $\bar{n}$ the behavior of our trial function
in $\bar{m}$ by looking at its derivative:
\begin{equation}
\partial_{\bar{m}}A_{trial}(\alpha,\beta,h_1,h_2)= 2 \beta \alpha
(1-\alpha)[\tanh (2\beta \alpha \bar{m} + h_1)-\bar{n}],
\end{equation}
so, at given $\bar{n}$ it is increasing in $\bar{m}$ and
$A_{trial}(\beta,h_1,h_2)$ convex in $\bar{m}$.
\newline
By a direct calculation we get the same result inverting $\bar{n}
\Leftrightarrow \bar{m}$:
\begin{equation}
\partial_{\bar{n}}A_{trial}(\alpha,\beta,h_1,h_2)= 2 \beta \alpha
(1-\alpha)[\tanh (2\beta (1-\alpha) \bar{n} + h_2)-\bar{m}].
\end{equation}
As it is crystal clear that the roles of $\bar{m},\bar{n}$ are of
the local magnetizations, we can allow ourselves in considering
only values $\bar{m} \geq \tanh(\beta h_1)$ such that there exist
a unique $\bar{n}(\bar{m})\geq 0 :
\tanh[2\beta(1-\alpha)\bar{n}(\bar{m})+h_1]=\bar{m}$.
\newline
From now on let us switch from $A_{trial}(\alpha,\beta,h_1,h_2)$
to a less rigorous labeling $\tilde{A}_{trial}(\bar{m},\bar{n})$
which aims to stress the relevant dependence by its variables time
by time: For this value $\bar{n}(\bar{m})$ lastly obtained, the
trial function has its minimum in $\bar{n}$ at fixed $\bar{m}$
such that we can substitute it and get
$\tilde{A}_{trial}(\bar{m})=\tilde{A}_{trial}(\bar{m},
\bar{n}(\bar{m}))$.
\newline
Now, as
\begin{equation}
\partial_{\bar{m}}\tilde{A}_{trial}(\bar{m}) = 2\beta \alpha
(1-\alpha)[\tanh(2\beta \alpha \bar{m} - h_1) - \bar{n}(\bar{m})],
\end{equation}
we can consider $\tilde{A}(\bar{m})$ as a function of $\bar{m}^2$
to see easily that $\tilde{A}$ is concave in $\bar{m}^2$, so it
has its unique maximum where its derivative vanishes.
\newline
Overall we can state that $A(\alpha,\beta,h_1,h_2)=
\sup_{\bar{m}}\inf_{\bar{n}}\tilde{A}(\bar{m},\bar{n})$, whose
stationary point is uniquely defined by the solution of the system
of self-consistence relations
\begin{eqnarray}
\bar m&=&\tanh\left( 2\beta(1-\alpha)\bar n + h_1 \right),\\
\bar n&=&\tanh\left( 2\beta\alpha\bar m + h_2 \right),
\end{eqnarray}
as stated in Theorem $1$.
\newline
Furthermore, if we restrict ourselves in considering
$\bar{n}=\bar{n}(\bar{m})$ -as imposed by the variational
principle- the error term (the fluctuation source) results
positive defined: This statement can be understood by
marginalizing with respect to the $\tau$ party the free energy
(summing over all the $\tau$-configurations) so to substitute
$n(\tau)$ by $\tanh[2\beta \alpha m(\sigma) t + 2\beta \alpha
\bar{m}(1-t) + h_1]$ and noting that $m(\sigma) \geq \bar{m}$
implies $\tanh[2\beta \alpha m t + 2\beta \alpha \bar{m}(1-t) +
h_1] \geq \bar{n}(\bar{m})$ such that
$$
\langle (m - \bar{m})[\tanh(2\alpha\beta\bar{m}+h_2)-\bar{n}] \geq
0,
$$
the error term is positive defined.
\newline
Now, in order to show that the error term is zero in the
thermodynamic limit (such that the expression of the trial becomes
correct) we proceed on a different way: the idea is to marginalize
with respect to one party, so to remain with a single
ferromagnetic party with a modified external field and then use
the standard package of knowledge developed for this case.
\newline
So, at first, we marginalize the free energy with respect to the
$\tau$ variables:
\begin{eqnarray}
A_{N_1, N_2}(\beta, h_1, h_2)&=&\frac{1}{N_1+N_2}\log Z_{N_1,
N_2}(\beta, h_1, h_2) = \\ &=&
\frac{1}{N_1+N_2}\log\sum_{\sigma}2^{N_2}\cosh^{N_2}(2\alpha\beta
m + h_2)\exp(h_1\sum_i\sigma_i) = \\ &=& (1-\alpha)\log 2 +
\frac{1}{N_1+N_2}\log\sum_{\sigma}e^{N_2 \log \cosh (2\beta\alpha
m+h_2)+h_2\sum_i\sigma_i}.
\end{eqnarray}
We can use now the convexity of the logarithm of the hyperbolic
cosine as $\log\cosh(x) \geq
\log\cosh(\bar{x})+(x-\bar{x})\tanh(\bar{x})$ to get
$$
e^{-N_2\log\cosh(2\alpha\beta m
+h_2)}e^{K_2\log\cosh(2\alpha\beta\bar{m}+h_2)}+N_2
2\alpha\beta(m-\bar{m})\tanh(2\alpha\beta\bar{m}+h_2)\leq 1,
$$
by which we bound the free energy through a new trial function
$\hat{A}_{trial}$ as
\begin{eqnarray}
A_{N_1, N_2}(\beta, h_1, h_2) &=& (1-\alpha)\log 2 +
\frac{1}{N_1+N_2}\log\sum_{\sigma}e^{N_2 \log \cosh (2\beta\alpha
m+h_2)+h_2\sum_i\sigma_i} \geq \\ \nonumber &\geq& (1-\alpha)\log
2 +
\frac{1}{N_1+N_2}\log\sum_{\sigma}\exp(N_2\log\cosh(2\alpha\beta\bar{m}+h_2))\times
\\ &\times& \exp(-2\alpha\beta N_2 \bar{m}\tanh(2\alpha \beta
\bar{m}+h_2))\exp(2\alpha\beta\tanh(2\alpha \beta
\bar{m}+h_2)\sum_i \sigma_i)=\\ \nonumber &=& \nonumber
(1-\alpha)(\log 2 + \log\cosh(2\alpha\beta \bar{m}+h_2))-2 \alpha
(1-\alpha) \beta \bar{m}\tanh(2\alpha \beta \bar{m}+h_2) +
\\ &+& \alpha \log 2 + \alpha \log\cosh[2\beta
(1-\alpha)\tanh(2\alpha\beta \bar{m}+h_2)]=\hat{A}_{trial}.
\nonumber
\end{eqnarray}
Once defined $\tilde{n}(\bar{m})=\tanh(2\alpha\beta \bar{m}+h_2)$,
we can look for the $\bar{m}$ streaming of the trial, namely
\begin{equation}
\partial_{\bar{m}}\hat{A}_{trial}(\bar{m})= (2 \beta
\alpha)^2(1-\alpha)(1-\tilde{n}^2)[\tanh(2\beta(1-\alpha)\tilde{n})-\bar{m}].
\end{equation}
If we now consider the streaming with respect to $\tilde{n}$ of
$\hat{A}_{trial}$ we get
\begin{equation}
\partial{\tilde{n}}\hat{A}_{trial}(\bar{m})=
(\partial_{\tilde{n}}\bar{m})\partial_{\bar{m}}\hat{A}_{trial}(\bar{m})=2\beta\alpha(1-\alpha)[\tanh(2\beta(1-\alpha)\tilde{n})-\bar{m}(\tilde{n})].
\end{equation}
So the streaming is decreasing in $\tilde{n}$ and the trial is
concave in $\tilde{n}$, there exist a unique maximum where the
derivative vanishes: Properly choosing $\bar{m}\to \bar{m}$ and
$\tilde{n}\to \bar{n}$ we get the statement of the theorem.
\end{proof}

As we are going to deepen the extremization  procedure in these
bipartite models, we stress that an important feature that seems
to arise from our study, is the occurrence of a $\min\max$
principle for the free energy, usually given by a maximum
principle in the ordered models, and a minimum principle in the
frustrated ones.
\newline
We finally report an interesting result  about the form of the
pressure (or equivalently, the free energy) of the model. Indeed
it turns out to be written as the convex combination of the
pressures of two different monopartite CW models, at different
inverse temperatures, as stated by the following

\begin{proposition}
In the thermodynamic limit, the following decomposition of bipartite ferromagnetic model free energies into convex sums of monoparty ones  is allowed:
$$
A(\alpha, \beta, h_1, h_2)=\alpha A^{CW}(\beta', h_1)+(1-\alpha)A^{CW}(\beta'', h_2),
$$
with $\beta'=2\beta(1-\alpha)\frac{\bar n}{\bar n}$ and $\beta''=2\beta\alpha\frac{\bar n}{\bar m}$.
\end{proposition}
\begin{proof}
We start setting the trial values for the inverse temperatures of the two monopartite models, as $\beta'=2\beta(1-\alpha)a^2$ and $\beta''=2\beta\alpha/a^2$, with $a$ a real parameter to be determined later, and we set the external fields to zero for the sake of clearness. It is
$$
Z_{N_1, N_2}(\beta)=\sum_{\sigma}\sum_{\tau}\exp\left(2\beta N\alpha(1-\alpha)mn\right),
$$
and since $2mn\leq m^2a^2+\frac{n^2}{a^2}$, $\forall a\neq0$, we have
\begin{eqnarray}
Z_{N_1, N_2}(\beta)&\leq&\sum_{\sigma}\sum_{\tau}\exp\left(2\beta N\alpha(1-\alpha)\frac{m^2a^2}{2}+2\beta N\alpha(1-\alpha)\frac{n^2}{2a^2}\right)\nonumber\\
&=&\sum_{\sigma}\exp\left(\beta' N_1\frac{m^2}{2}\right)\sum_{\tau}\exp\left(\beta'' N_2\frac{n^2}{2}\right)\nonumber\\
&=&Z_{N_1}(\beta')Z_{N_2}(\beta'')\nonumber,
\end{eqnarray}
thus we conclude, when $N\to\infty$
\begin{equation}\label{eq:conv<}
A(\alpha, \beta)\leq \alpha A^{CW}(\beta')+(1-\alpha)A^{CW}(\beta'')
\end{equation}
The inverse bound is proven noticing that
\begin{eqnarray}
\frac{Z_{N_1, N_2}(\beta)} {Z_{N_1}(\beta')Z_{N_2}(\beta'')}&=&\frac{\sum_{\sigma}\sum_{\tau}\exp\left(2\beta N\alpha(1-\alpha)mn-\frac{\beta'm^2}{2}-\frac{\beta''n^2}{2}\right)\exp\left( \frac{\beta'm^2}{2}+\frac{\beta''n^2}{2} \right)}{Z_{N_1}(\beta')Z_{N_2}(\beta'')}\nonumber\\
&=&\Omega\left[\exp\left( 2\beta N\alpha(1-\alpha)mn-\frac{\beta'm^2}{2}-\frac{\beta''n^2}{2} \right)\right]\nonumber\\
&\geq&\exp\left(\Omega\left[ 2\beta N\alpha(1-\alpha)mn-\frac{\beta'm^2}{2}-\frac{\beta''n^2}{2} \right]\right)\nonumber\\
&=&\exp\left( 2\beta N\alpha(1-\alpha)\bar m\bar n-\frac{\beta'\bar m^2}{2}-\frac{\beta''\bar n^2}{2} \right)\nonumber,
\end{eqnarray}
where we have denoted with $\Omega$ the joint state of the two monopartite systems. Now, bearing in mind the definition of $\beta'$ and $\beta''$ given at the beginning, we get
$$
\frac{Z_{N_1, N_2}(\beta)}{Z_{N_1}(\beta')Z_{N_2}(\beta'')}\geq e^{-2N\beta\alpha(1-\alpha)\left(a\bar m-\frac{\bar n}{a} \right)^2}
$$
and then, in thermodynamic limit
\begin{equation}\label{eq:conv.CW2}
A(\alpha, \beta)\geq\alpha A^{CW}(\beta')+(1-\alpha)A^{CW}(\beta'')-2\beta\alpha(1-\alpha)(\left(a\bar m-\frac{\bar n}{a} \right)^2)
\end{equation}
Now we notice that the extrema in (\ref{eq:conv<}) and
(\ref{eq:conv.CW2}),  respectively a minimum and a maximum, are
obtained with the choice $a^2=\frac{\bar n}{\bar m}$, that
completes the proof.
\end{proof}
It is worthwhile to remark that the two monoparty models here
are completely independent, in the sense that the order parameters
are given by the relations
\begin{eqnarray}
\bar m&=&\tanh (\beta' \bar m),\nonumber\\
\bar n&=&\tanh (\beta'' \bar n).\nonumber
\end{eqnarray}

\subsection{The free energy again: a maximum principle}
Our aim is now to front the mathematical study of this model with
the approach described in
\cite{sum-rules}\cite{barra0}\cite{io1}\cite{dibiasio}, based on a
mechanical interpretation of the interpolation method. The problem
of finding the free energy in the thermodynamic limit is here
translated in solving an Hamilton-Jacobi equation with certain
suitable boundary conditions, and an associated Burgers transport
equation for the order parameter of the model. In order to
reproduce this scheme, with the freedom of interpretation of the
label $t$ for the time and $x$ for the space, let us introduce now
the $(x,t)$-dependent interpolating partition function
\begin{eqnarray}
Z_N(x,t)&=& \sum_{\sigma}\sum_{\tau}\exp N \Big(t\alpha_N(1-\alpha_N)m_Nn_N+\nonumber\\
&+&\frac{(\beta-t)}{2}(a^2\alpha^2m_N^2+\left(\frac{1-\alpha}{a}\right)^2n_N^2)+x(a\alpha_N m_N-\frac{1-\alpha_N}{a}n_N)\nonumber\\
&+& h_1\alpha_N m_N+h_2(1-\alpha_N)n_N \Big)\nonumber,
\end{eqnarray}
such that the thermodynamical partition function of the model is
recovered when $t=2\beta$ and $x=0$. At this level $a$ is a free parameter to be determined later. We can go further and
explicitly define the function
\begin{equation}
\varphi_N(x,t)=\frac{1}{N}\log Z_N(x,t),
\end{equation}
that therefore is just the pressure of the model for a suitable
choice of $(x,t)$. Now, computing derivatives of $\varphi_N(x,t)$,
we notice that, putting $D_N=a\alpha_N m_N-\frac{1-\alpha_N}{a}n_N$, it is
\begin{eqnarray}
\partial_t\varphi_N(x,t)&=&-\frac{1}{2}\meanv{D^2_N}(x,t),\nonumber\\
\partial_x\varphi_N(x,t)&=&\meanv{D_N}(x,t),\nonumber\\
\partial^2_x\varphi_N(x,t)&=&\frac{N}{2}\left( \meanv{D_N^2}-\meanv{D_N}^2 \right).\nonumber
\end{eqnarray}
Thus we can build our differential problems trough an Hamilton-Jacobi
equation for $\varphi_N(x,t)$
\begin{equation}\label{eq:HJ-N-bip}
\left\{
\begin{array}{rclll}
&& \partial_t\varphi_N(x,t)+\frac{1}{2}(\partial_x\varphi_N(x,t))^2+\frac{1}{2N}\partial^2_{x}\varphi_N(x,t)= 0 &\:&\mbox{in }{\mathbb R}\times(0,+\infty)\\
&& \varphi_N(x,0)=\alpha_N A_{N_1}^{CW}(\beta', \alpha(h_1+x)) + (1-\alpha_N)
A_{N_2}^{CW}(\beta'', (1-\alpha)(h_2-x))&\:&\mbox{on  }{\mathbb R}\times
\{t=0\},
\end{array}
\right.
\end{equation}
where $A^{CW}_{N_1}$ is the pressure of the Curie-Weiss model made of
by $N_1$ $\sigma$ spins with inverse temperature $\beta'$, and $A^{CW}_{N_2}$ is the same referred to
the $N_2$ $\tau$ spin with inverse temperature $\beta''$, with $\beta'=2\beta a^2(1-\alpha)$ and $\beta''=2\beta a^{-2}\alpha$,
and trough a Burgers equation for the velocity field $D_N(x,t)$
\begin{equation}\label{eq:burgers-N-bip}
\left\{
\begin{array}{rclll}
&& \partial_t D_N(x,t)+D_N(x,t)\partial_x D_N(x,t)+\frac{1}{2N_1}\partial^2_{x}D_N(x,t)= 0&\:&\mbox{in }{\mathbb R}\times(0,+\infty)\\
&& D_N(x,0) = \alpha m(\beta, h_1+x)-(1-\alpha_N) n(\alpha_N\beta,
h_2-x)&\:&\mbox{on  }{\mathbb R}\times \{t=0\},
\end{array}
\right.
\end{equation}
This is true of course for every choice of the parameter $a$, that has the role of balancing the weights of the single party contributions. Since we have seen that the function $\varphi_N$ is decreasing in time, if we put $x=0$, with no external fields, we gain
$$
A_{N_1, N_2}(\beta)\leq \alpha_N A^{CW}_{N_1} (\beta')+(1-\alpha_N)A^{CW}_{N_2}(\beta'').
$$
If we take $a^2=\sqrt{\frac{1-\alpha}{\alpha}}$, such that it is $\beta'=\beta''=\bar\beta=2\beta\sqrt{\alpha(1-\alpha)}$, we have, in the infinite volume limit
$$
A(\beta)\leq \alpha A^{CW} (\bar\beta)+(1-\alpha)A^{CW}(\bar\beta),
$$
that easily give us the critical line of the bipartite model,
$2\beta\sqrt{\alpha(1-\alpha)}= 1$, obtained straightly by the
critical point of the two Curie-Weiss models, $\bar\beta=1$.
Anyway, for reasons that will become clear soon, hereafter it will
adopted the different value $a=1$.

\begin{remark}
We stress that the boundary condition in equation (\ref{eq:HJ-N-bip}) is always an upper bound for $\varphi_N$.
\end{remark}
In order to work out an explicit solution for the thermodynamic
limit of the pressure, in primis we notice that the main
difference with respect to the single party (namely the Curie
Weiss \cite{io1}) is the more delicate form of the boundary
conditions. In fact we have that interactions do not factorize
trivially (in a way independent by the size of the system). It is
\begin{equation}\label{eq:bordo1}
\varphi_N(x,0)=\alpha_NA^{CW}_{N_1}(\beta', h_1+x) +(1- \alpha_N)
A_{N_2}^{CW}(\beta'', h_2-x),
\end{equation}
however, it is known \cite{io1} how to get a perfect control on
the function on the r.h.s. of (\ref{eq:bordo1}), and we have
\begin{equation}\label{eq:bordo2}
\varphi_N(x,0)=\alpha A^{CW}(\beta', h_1+x)+(1-\alpha) A^{CW}(\beta'',
h_2-x)+\OO{N}.\end{equation} Evenly we have for the velocity field in $t=0$
$$
\alpha_N m_{N_1}(\beta', h_1+x)-(1-\alpha_N) n_{N_2}(\beta'', h_2-x)=\alpha m(\beta',
h_1+x)-(1-\alpha) n(\beta'', h_2-x)+\OO{\sqrt{N}}.
$$
\newline
%-\frac{\beta^2}{2}M(\beta, h+x)+\log\mathbb{E}_{\sigma}\left[ \exp((x+h+\beta M(\beta, h+x))\sigma) \right]
We obtained an Hamilton-Jacobi equation for the free
energy with a vanishing dissipative term in the
thermodynamic limit, while the velocity field, that is the order parameter, satisfies a Burgers' equation with a mollifier
dissipative term.
\newline
We stress that our method introduces by itself the correct
order parameter, without imposing it by hands.
\begin{remark}
As the next condition on $\varphi_N(x,t)$ is tacitely required by
the following lemma, we stress that the function $D_N(x,t)$ is
bounded uniformly in $N, \alpha, \beta, h_1, h_2$, that implies
the function $\varphi_N(x,t)$ to be Lipschitz continuous.
\end{remark}

We can replace the sequence of differential problems with boundary
conditions dependent by $N$ with the same sequence of equations but
with obvious fixed boundary conditions, that is the well defined limiting
value of $\varphi_N$ and $D_N$ in $t=0$. To this purpose
it is useful the following
\begin{lemma}
The two differential problems
\begin{equation}\label{eq:HJ-N-bip_FIN}
\left\{
\begin{array}{rclll}
&& \partial_t\varphi_N(x,t)+\frac{1}{2}(\partial_x\varphi_N(x,t))^2+\frac{1}{2N_1}\partial^2_{x}\varphi_N(x,t)= 0 &\:&\mbox{in }{\mathbb R}\times(0,+\infty)\\
&& \varphi_N(x,0)=\alpha_N A_{N_1}^{CW}(\beta', h_1+x) + (1-\alpha_N)
A_{N_2}^{CW}(\beta'', h_2-x)=h_N(x)&\:&\mbox{on  }{\mathbb R}\times
\{t=0\},
\end{array}
\right.
\end{equation}
and
\begin{equation}\label{eq:HJ-N-bip_FIN}
\left\{
\begin{array}{rclll}
&& \partial_t\bar\varphi_N(x,t)+\frac{1}{2}(\partial_x\bar\varphi_N(x,t))^2+\frac{1}{2N_1}\partial^2_{x}\bar\varphi_N(x,t)= 0 &\:&\mbox{in }{\mathbb R}\times(0,+\infty)\\
&& \varphi_N(x,0)=\alpha A^{CW}(\beta', h_1+x) + (1-\alpha)
A^{CW}(\beta'', h_2-x)=h(x)&\:&\mbox{on  }{\mathbb R}\times
\{t=0\},
\end{array}
\right.
\end{equation}
are completely equivalent, \textit{i.e.} in thermodynamic limit
they have the same solution, $\varphi_N\to\varphi$ and
$\bar\varphi_N\to\varphi$ and it is
$$
|\varphi_N-\bar\varphi_N|\leq\OO{N}.
$$
\end{lemma}
\begin{proof}
By a Cole-Hopf transform, we can easily write the general form of
$\delta_N(x,t)=|\varphi_N(x,t)-\bar\varphi_N(x,t)|$ as
$$
\delta_N=\frac{1}{N}\left|\log\frac{\int_{-\infty}^{+\infty}dy\Delta(y,(x,t))e^{-NR_N(y)}}{\int_{-\infty}^{+\infty}dy\Delta(y,(x,t))}\right|,
$$
where we introduced the modified heat kernel
$\Delta(y,(x,t))=\sqrt{\frac{N}{2\pi t}}\exp\left(-N\left[
(x-y)^2/2t+h(y)\right]\right)$, and $R_N(y)=|h(y)-h_N(y)|$, with
$\lim_N NR_N<\infty$, $\forall\:y$. Now we notice that as it
certainly exists an $y^*$ such that
$$
\sup_y R_N(y)=y^*\qquad\mbox{ and }\qquad\lim_N NR_N(y^*)<\infty,
$$
hence it is
\begin{equation}
\delta_N(x,t)\leq \frac{1}{N} |\log e^{-NR_N(y^*)}| = \frac{1}{N}\left[NR_N(y^*)\right]\leq\OO{N},
\end{equation}
that completes the proof.
\end{proof}
Of course a similar result holds also for the Burgers' equation for the velocity field $D_N$.
\newline
Now the path to follow is clear: the problem of existence and
uniqueness of the thermodynamic limit is translated here into the
convergence of the viscous mechanical problem to the free one. We
can readapt a theorem that resumes a certain amount of results due
to Douglis, Hopf, Lax and Oleinik \cite{lax}\cite{io1} which
assures the existence and uniqueness of the solution to the free
problem:
\begin{theorem}
The pressure of the generalized bipartite ferromagnet, in
 the thermodynamic limit, exists, is unique and is given by:
\begin{equation}\label{eq:A_BIP}
A(\beta, \alpha, h_1, h_2)=-2\alpha(1-\alpha)\beta\bar{n}\bar{m} +
\alpha\log\cosh\left(h_1+2(1-\alpha)\beta \bar{n}\right)+(1-\alpha)\log\cosh\left(h_2+2\alpha\beta
\bar{m}\right),
\end{equation}
where, given the well defined magnetization for the generalized CW
model respectively for $\sigma$ and $\tau$, $m(\beta, h)$ and
$n(\beta, h)$, it is
\begin{eqnarray}
\tilde m(\beta, \alpha, h_1, h_2)&=& m(2\beta, h_1-D)\label{eq:Mtilde}\\
\tilde n(\beta, \alpha, h_1, h_2)&=&n(2\beta, h_2+D)\label{eq:Ntilde}.
\end{eqnarray}
Furthermore it is
\begin{equation}\label{eq:conv_bip}
|A_N(\beta, h_1, h_2)-A(\beta, \alpha, h_1, h_2)|\leq\OO{N}.
\end{equation}
\end{theorem}
\begin{proof}
Well known results about CW model (see for instance
\cite{barra0}\cite{io1}) give us the existence and the form of the
free solution. We know \cite{io1} that the free Burger's equation
can be solved along the characteristics
\begin{equation}\label{eq:traiettorie-generali_BIP}
\left\{
\begin{array}{rrl}
t&=&s\\
x&=&x_0+sD(x_0, 0),
\end{array}
\right.
\end{equation}
where
$$
D(x_0, 0)=\alpha m(2\beta(1-\alpha), \alpha(h_1+x_0))+(1-\alpha) n(2\beta\alpha, (1-\alpha)(h_2-x_0)),
$$
and it is
\begin{eqnarray}
D(x,t)&=&D(x_0(x,t),0) \\ \nonumber &=& \alpha m(2\beta(1-\alpha),
\alpha(h_1+x-tD(x_0,0)))+(1-\alpha) n(2\beta\alpha,
(1-\alpha)(h_2-x+tD(x_0,0))).
\end{eqnarray}
Then we can notice that
\begin{eqnarray}
m(\beta', h_1-x+tD(x_0,0))&=& \tanh\left(h_1-x+t(1-\alpha) n\right)\label{eq:self-M-bip},\\
n(\beta'',h_2+x-tD(x_0,0))&=&\tanh\left(h_2+x+t\alpha m\right)\label{eq:self-N-bip},
\end{eqnarray}
which coincide with (\ref{eq:Mtilde}) and (\ref{eq:Ntilde}) when
$x=0$ and $t=2\beta$.
At this point we know \cite{io1} that the minimum is taken for
$y=x-tD(x,t)$, such that we have
\begin{eqnarray}
\varphi(x,t)_{(x=0,t=2\beta)}&=&\Big[\frac{t}{2}D^2(x,t)-\frac{t}{2}\alpha^2m^2(\beta', h_1+x-tD(x_0,0))-\frac{t}{2}(1-\alpha)^2 n^2(\beta'', h_2-x \nonumber \\ &+& tD(x_0,0))
+ \alpha\log\cosh\left(h_1-x+m(2\beta(1-\alpha)+t\alpha)-t(1-\alpha) n\right)\nonumber\\
&+&(1-\alpha)\log\cosh\left(h_2+x+t\alpha m-n((1-\alpha)t-2\beta\alpha)\right)\Big]_{(x=0,t=2\beta)}\nonumber\\
&=&A(\beta, \alpha, h_1, h_2),\nonumber
\end{eqnarray}
where $A(\beta, \alpha, h_1, h_2)$ is  just given by
(\ref{eq:A_BIP}), bearing in mind the right definition of $\tilde
M$ and $\tilde N$. Now we must only prove the convergence of the
true solution to the free one. But equation (\ref{eq:conv_bip})
follows by standard techniques, because of the uniform concavity
of
$$
\frac{(x-y)^2}{2t}+\alpha A^{CW}(\beta, h_1+y)+(1-\alpha) A^{CW}(\beta,
h_2-y)
$$
with respect to $y$. In fact we have that, by a Cole-Hopf
transform, the unique bounded solution of the viscous problem is
$$
\varphi_N(x,t)=\frac{1}{N}\log\sqrt{\frac{N}{t}}\int
\frac{dy}{\sqrt{2\pi}}\exp\left[-N\left(
\frac{(x-y)^2}{2t}+\alpha A^{CW}(\beta, h_1+y)+(1-\alpha) A^{CW}(\beta,
h_2-y) \right)\right]
$$
and we have, by standard estimates of a Gaussian integral, that
$$
\left|\varphi(x,t)-\varphi_N(x,t)\right|\leq\OO{N},
$$
\textit{i.e.} also eq. (\ref{eq:conv_bip}) is proven.
\end{proof}
%Finally, by this theorem, we can easily write down the free energy
%of the model:
%$$
%f(\alpha, \beta, h_1, h_2)=2\alpha(1-\alpha)\tilde N(D)\tilde M(D) -
%\frac{\alpha}{\beta}\log\cosh\left(h_1+2\beta(1-\alpha)
%\tilde N\right)-\frac{1-\alpha}{\beta}\log\cosh\left(h_2+2\beta\alpha
%\tilde M\right).
%$$
It is interesting to notice that here the minimax principle
discussed in the previous section has become a pure maximum
principle for the free energy, because of the natural choice of
the order parameter $D$, that is in our formalism the analogue of
the velocity field. Thus we have outlined the framework for
stating the next
\begin{proposition}
As an alternative to Theorem $1$, the free energy of the bipartite
ferromagnet can be obtained even within a classical extremization
procedure as it is uniquely given by the following variational
principle
\begin{eqnarray}
f(\alpha, \beta, h_1, h_2)&=&\max_{D}\Big[ - D^2+\alpha^2m^2(D)+(1-\alpha)^2n(D)^2\nonumber\\
&-&\frac{\alpha}{\beta}\log\cosh\left( h_1+\beta(1-\alpha) n(D
\right) - \frac{1-\alpha}{\beta}\log\cosh\left( h_2+\beta\alpha m(D)
\right) \Big].\nonumber
\end{eqnarray}
\end{proposition}

\section{Bipartite spin glasses}

Let us consider a set of $N_1$ i.i.d. random spin variables
$\sigma_i$, $i=1,..., N_1$ and let us consider also
another set of  i.i.d. random spin variable $\tau_{j}$,
$j=1,...,N_2=N-N_1$. We will consider for sake of simplicity only
dichotomic spin variables, although our scheme is easily
extensible to other spin distributions, symmetric and with compact
support. Therefore we have two distinct sets (or parties
hereafter) of different spin variables, and  we let them interact
through the following Hamiltonian:
\begin{equation}\label{eq:H1}
H_{N_1,N_2}(\sigma, \tau)=-\sqrt{\frac{2}{N_1+N_2}}\sum_{i=1}^{N_1}\sum_{j=1}^{N_2} \xi_{ij}\sigma_i\tau_j,
\end{equation}
where the $\xi_{ij}$ are also i.i.d. r.v., with $\mathbb{E}[\xi]=0$ and $\mathbb{E}[\xi^2]=1$,  \textit{i.e.} the quenched noise ruling the mutual interaction between parties. In particular we deal with a $\mathcal{N}(0,1)$ quenched disorder. It is then defined a bipartite model of spin glass where emphasis is given on its bipartite nature by neglecting self-interactions, mirroring the strategy outlined in the first part of the work.
\newline
For the sake of simplicity, as external fields in complex  systems
must be considered much more carefully with respect to the simple
counterparts, we are going to work out the theory neglecting them
at this stage.

Of course, once the Hamiltonian is given, it results defined the
partition function, the pressure and the free energy of the model,
as
\begin{eqnarray}
Z_{N_1, N_2}(\beta)&=&\sum_{\sigma,\tau}\exp\left(-\beta H_{N_1,N_2}(\sigma, \tau)\right),\\
A_{N_1, N_2}(\beta)&=&\frac{1}{N+K}\mathbb{E}_J \log Z_{N_1, N_2}(\beta),\\
f_{N_1, N_2}(\beta)&=&-\frac{1}{\beta}A_{N_1, N_2}(\beta).
\end{eqnarray}
We can define also the Boltzmann state for an observable function of the spin variables $z(\sigma, \tau)$:
$$
\omega_{N_1, N_2}(z)=Z^{-1}_{N_1,
N_2}(\beta)\sum_{\sigma,\tau}\left[z(\sigma, \tau)\exp\left(-\beta
H_{N_1,N_2}(\sigma, \tau)\right)\right].
$$
and, as in glasses we need to introduce replicas (equivalent
copies of the system with the same identical quenched disorder),
we can define even the Boltzmann product state as
$\Omega_{N_1,N_2}=\omega_{N_1,N_2}\times...\times\omega_{N_1,N_2}$, where the
amount of replicas can be specified time by time.
\begin{remark}
In analogy with the prescription introduced in  the normalization
of the bipartite ferromagnet, the factor $\sqrt{2}$ in the
Hamiltonian is put ad hoc in order to obtain for the balanced
bipartite spin glass $(\alpha=1/2)$ the same critical point of the
Sherrington- Kirkpatrick one party model, as will be clarified in
the next section.
\end{remark}

The main achievement of the theory would be a complete control of
the free energy, or the pressure, in the thermodynamic limit,
\textit{i.e.} for $N_1,N_2\to\infty$. We stress that two different
cases may arise: the first is that the size of one party grows
faster than the other; the second is that the two sizes grows in
the same way, such that is well defined the ratio $N_1/N\to
\alpha\in(0, 1)$ and $N_2/N\to (1-\alpha)\in(0, 1)$ again for
coherence with the strategy outlined in the first part of the work
and for a general higher interest in this case. We will adopt this
latter definition of thermodynamic limit, and thus the
thermodynamic functions depend by the additional parameter
$\alpha$ ruling the relative ratio among the parties:
$$
\lim_N A_{N_1,N_2}(\beta)=A(\alpha, \beta).
$$
We must stress that at the moment no rigorous proof of the
existence of such a limit is known and that there is a deep
connection among this limit and the one of the Hopfield model for
neural networks \cite{BG1}\cite{BGG1}. Finally we must introduce
overlap, that is correlation functions among replicas. It
naturally arises how in this model we have two kind of such a
quantity, one referred to each party; in fact we define
immediately
\begin{eqnarray}
q_{ab}&=&\frac{1}{N_1}\sum_{i}\sigma^a_i\sigma^b_i,\nonumber\\
p_{ab}&=&\frac{1}{N_2}\sum_{j}\tau^a_j\tau^b_j\nonumber.
\end{eqnarray}

\subsection{High temperature behavior}

We start our study of the model by characterizing the high temperature regime. It turns out that the system behaves like the annealed one in a wide region of the $(\alpha, \beta)$-plane, as stated by the following

\begin{theorem}
The pressure of the bipartite spin glass model does coincide with its annealed one
\begin{equation}\label{eq_A-ann}
A_A(\alpha, \beta)=\log 2 +\beta^2\alpha(1-\alpha)
\end{equation}
in the region of the $(\alpha, \beta)$ plane defined by
$2\beta^2\sqrt{\alpha(1-\alpha)}\leq1$.
\end{theorem}
\begin{proof}
At first we calculate the annealing:
\begin{equation}\label{eq:EZ}
\mathbb{E}[Z_{N_1,N_2}(\beta)]=\exp(N)\left(\log2 + \beta^2\alpha(1-\alpha) \right),
\end{equation}
thus
\begin{equation}
A_A(\alpha, \beta) = \lim_{\begin{matrix} N_1,N_2 \\ \frac{N_1}{N}\to\alpha \end{matrix}} (N)^{-1}\log\mathbb{E}[Z_{N_1,N_2}(\beta)] = \log 2 +\beta^2\alpha(1-\alpha).\nonumber
\end{equation}
Now, following a standard method \cite{talbook}\cite{BG1}, we want to use the Borel-Cantelli lemma on $Z_{N_1,N_2}/\mathbb[Z_{N_1,N_2}]$. We evaluate the second moment of the partition function:
\begin{eqnarray}
\mathbb{E}[Z^2_{N_1,N_2}]&=&\mathbb{E}_{\xi}\sum_{\{\sigma,\tau\}} \exp\left( \beta\sqrt{\frac{2}{N}} \sum_{ij} \xi_ij(\sigma_i\tau_j\sigma'_i\tau'_j)\right)\nonumber\\
&=&\sum_{\{\sigma,\tau\}} \exp\left( \beta^2\frac{N_1 N_2}{N}(2+q_{12}p_{12}) \right)\nonumber\\
&=&e^{2(N)(\log2+\beta^2\alpha(1-\alpha))}\sum_{\{\sigma,\tau\}} \exp \left((N)2\beta^2\alpha(1-\alpha)q_{12}p_{12}\right),
\end{eqnarray}
where we neglected terms leading to an error for the pressure $O(N^{-2})$, replacing $\alpha_N$ with $\alpha$. Now we perform the transformation $\sigma\to\sigma\sigma'$ and $\tau\to \tau\tau'$ in order to get $q_{12}\to m$ and $p_{12}\to n$. Thus we have
\begin{eqnarray}
\frac{\mathbb{E}[Z^2_{N_1,N_2}]}{\mathbb{E}[Z_{N_1,N_2}]^2}&=&\sum_{\{\sigma,\tau\}} \exp \left((N_1+N_2)2\beta^2\alpha(1-\alpha)mn\right)  \nonumber\\
&=&\sum_{\{\sigma,\tau\}} \exp \left((N_1+N_2)\beta'\alpha(1-\alpha)mn\right),\nonumber\\
\end{eqnarray}
with $\beta'=\beta^2$. The last term, as we have seen in the
previous section  about bipartite ferromagnetic models, is bounded
for $2\beta'\sqrt{\alpha(1-\alpha)}\leq 1$, \textit{i.e.}
$2\beta^2\sqrt{\alpha(1-\alpha)}\leq1$, that completes the proof.
\end{proof}

We notice that this is a result slightly different with respect to
the Hopfield model \cite{BG1}. In fact we have that the annealed
free energy and the true one are the same still at small
temperatures, depending on the different weights assumed by the
parties (that are of course ruled by $\alpha$). This is reflected
by the symmetry of the high temperature region with respect to the
line $\alpha=1/2$. The following argument is going to clarify this
point:
\newline
Given respectively a $N_1\times N_1$ and a $N_2\times N_2$ random matrix $J$ and $J'$, with both $J_{ij}$ and $J'_{ij}$ normal distributed random variables for every $i,j$, we introduce the interpolating partition function
$$
Z_{N_1,N_2}(\beta, t)=\sum_{\sigma, \tau} e^{\left(\beta\sqrt{\frac{2t}{N}}\sum^{N_1,N_2}_{ij}\xi_{ij}\sigma_i\tau_j+a^2\beta\sqrt{\frac{2(1-t)}{N}}\sum^{N_1}_{ij}J_{ij}\sigma_i\sigma_j+\frac{\beta}{a^2}\sqrt{\frac{2(1-t)}{N}}\sum^{N_2}_{ij}J'_{ij}\tau_i\tau_j\right)},
$$
where $a$ is a parameter to be determined a posteriori.
Putting $\beta'=\beta a^2\sqrt{2\alpha}$ and $\beta''=\beta a^{-2}\sqrt{2(1-\alpha)}$, and neglecting terms vanishing when $N$ grows to infinity, we can rewrite the latter expression as
$$
Z_{N_1,N_2}(\beta, \beta', \beta'', t)=\sum_{\sigma, \tau} e^{\left(\beta\sqrt{\frac{2t}{N}}\sum^{N_1,N_2}_{ij}\xi_{ij}\sigma_i\tau_j+\beta'\sqrt{\frac{(1-t)}{N_1}}\sum^{N_1}_{ij}J_{ij}\sigma_i\sigma_j+\beta''\sqrt{\frac{(1-t)}{N_2}}\sum^{N_2}_{ij}J'_{ij}\tau_i\tau_j\right)}.
$$
Now we introduce the function
$$
\phi_{N_1,N_2}(t,\beta, \beta', \beta'')=\frac{1}{N}\mathbb{E}\log Z_{N_1,N_2}(\beta, \beta', \beta'', t)+\frac{t}{4}\left(\beta'^2\alpha+\beta''^2(1-\alpha)+4\beta^2\alpha(1-\alpha)\right).
$$
It is easily verified that
\begin{equation}\label{eq:phi_Netstr}
\left\{
\begin{array}{rclll}
\phi_{N_1,N_2}(t=0)&=&\alpha A_{N_1}^{SK}(\beta')+(1-\alpha)A_{N_2}^{SK}(\beta'')\\
\phi_{N_1,N_2}(t=1)&=&A_{N_1,N_2}(\beta)+\frac{1}{4}\left(\beta'^2\alpha+\beta''^2(1-\alpha)-4\beta^2\alpha(1-\alpha)\right)\\
\end{array}
\right.
\end{equation}
where, of course, $A^{SK}$ is the pressure for the Sherrington-Kirkpatrick model, that is, the only one party model. Furthermore we can take the derivative in $t$ and obtain
$$
\frac{d}{dt}\phi_{N_1+N_2}=\meanv{\left( \beta'\sqrt{\alpha}q_{12}-\beta''\sqrt{1-\alpha}p_{12} \right)^2}\geq0,
$$
since $2\beta'\beta''\sqrt{\alpha(1-\alpha)}=4\beta\alpha(1-\alpha)$. Hence we get the following bound for the pressure
\begin{equation}\label{eq:ANK>ANAK-1}
A_{N_1,N_2}(\beta)\geq \alpha A_{N_1}^{SK}(\beta')+(1-\alpha)A_{N_2}^{SK}(\beta'')-\frac{1}{4}\left(\beta'^2\alpha+\beta''^2(1-\alpha)+4\beta^2\alpha(1-\alpha)\right).
\end{equation}
Now we can fix $a^2$ in such a way that $\beta'=\beta''=\bar\beta$. As a consequence, it results $a^4=\sqrt{(1-\alpha)/\alpha}$ and $\bar\beta^2=2\beta^2\sqrt{\alpha(1-\alpha)}$, and the formula (\ref{eq:ANK>ANAK-1}) becomes
\begin{equation}\label{eq:ANK>ANAK-2}
A_{N_1,N_2}(\beta)\geq \alpha A_{N_1}^{SK}(\bar\beta)+(1-\alpha)A_{N_2}^{SK}(\bar\beta)-\frac{\bar\beta^2}{4}+\beta^2\alpha(1-\alpha).
\end{equation}
Thus the pressure is always greater than the convex sum of the single party Sherrington-Kirkpatrick pressure. The extra term is build in such a way we get an equality in the annealed region. In fact we have, if $\bar\beta\leq1$, \textit{i.e.} $2\beta^2\sqrt{\alpha(1-\alpha)}\leq1$, that in thermodynamic limit both $A_{N_1}^{SK}$ and $A_{N_2}^{SK}$ are $\log2 + \bar\beta^2/4$, and therefore we get
$$
\log2+\beta^2\alpha(1-\alpha)\geq A(\alpha,\beta)\geq \log2+\beta^2\alpha(1-\alpha),
$$
where, as usual, the upper bound is given by the Jensen inequality.

\subsection{Replica symmetric free energy}

In order to obtain an explicit expression for the free energy
density (or equivalently the pressure $A(\alpha,\beta)$), we apply
the double stochastic stability technique recently developed in
\cite{BGG1}. In a nutshell the idea is to perturb stochastically
both the parties via random perturbations; these are coupled with
scalar parameters to be set a fortiori in order to get the desired
level of approximation. With these perturbations the calculations
can be reduced to a sum of one body problems via a suitable sum
rule for the free energy; by the latter, the replica symmetric
approximation can be obtained straightforwardly by neglecting the
fluctuations of the order parameters.
\newline
Concretely we introduce the following interpolating partition function, for $t\in[0,1]$
\begin{eqnarray}\label{inter}
Z_{N_1,N_2}(t) &=& \sum_{\sigma}\sum_{\tau}\exp(\sqrt{t}\frac{\sqrt{2}\beta}{\sqrt{N}}\sum_{i
j}^{N_1,N_2}\xi_{i,j} \sigma_i \tau_{j})\cdot \\
\nonumber &\cdot& \exp(\sqrt{1-t}[\beta \sqrt{2(1-\alpha)\bar p}
\sum_i^{N_1} \eta_i \sigma_i + \beta \sqrt{2\alpha \bar q} \sum_{j}^{N_2}
\tilde{\eta}_{j} \tau_{j}]),
\end{eqnarray}
where $\eta, \tilde{\eta}$ are stochastic perturbations, namely
i.i.d. random variables $\mathcal{N}[0,1]$, whose averages are
still encoded into $\mathbb{E}$, and, so far,  $\bar q$, $\bar p$
are Lagrange multipliers to be determined later.
\newline
Now we introduce the interpolating function
$$
A_{N_1,N_2}(t)=\frac{1}{N} \mathbb{E} \log Z_{N_1,N_2}(t)+(1-t)\alpha(1-\alpha)\beta^2(1-\bar q)(1-\bar p).
$$
It is easily seen that at $t=1$ we recover the original pressure $A(\alpha,\beta)$, while for $t=0$ we obtain
a factorized one-body problem:
\begin{equation}\label{eq:A_NK-etstremi}
\left\{
\begin{array}{rclll}
\lim_N A_{N_1,N_2}(t=1)&=&A_{N_1,N_2}(\beta),\\
\lim_N A_{N_1,N_2}(t=0)&=& A_0(\alpha,\beta)+\alpha(1-\alpha)\beta^2(1-\bar q)(1-\bar p),
\end{array}
\right.
\end{equation}
with
\begin{eqnarray} \nonumber
A_0(\alpha, \beta) &=& \frac{1}{N} \mathbb{E} \log \sum_{\sigma}
\exp(\beta\sqrt{2(1-\alpha)\bar p}\sum_i \eta_i \sigma_i) \nonumber\\
&+& \frac{1}{N} \mathbb{E}
\log \sum_{\tau} \exp (\beta \sqrt{2\alpha \bar q}
\sum_{j} \tilde{\eta}_{j} \tau_{j}) \\
\nonumber &=& \ln 2 + \alpha\mathbb{E}_g \log \cosh \Large(
g\beta\sqrt{2(1-\alpha)\bar{p}} \Large) + (1-\alpha) \mathbb{E}_g
\log \cosh \Large( g\beta\sqrt{2\alpha\bar{q}} \Large),
\end{eqnarray}
where $\mathbb{E}_g$ indicates the expectation with respect to the $\mathcal{N}(0,1)$ r.v. $g$.
\newline
Now we must evaluate the $t$-derivative of $A_{N_1,N_2}(t)$ in order
to get a sum rule, namely \be A(t=1) = A_0(\alpha,\beta) +
\int_0^1 dt \left(\frac{d}{dt} A(t)\right). \ee Denoting via
$\langle \rangle_t$ the extended Boltzmann measure encoded into
the structure (\ref{inter}) -that reduces to the standard one for
$t=1$ as it should-, we get three terms by deriving the four
contributions into the extended Maxwell-Boltzmann exponential, that
we call $\mathcal{A}, \mathcal{B}, \mathcal{C}$ and follow:

\begin{eqnarray}
\mathcal{A} &=& \frac1N
\mathbb{E}\frac{\beta}{2\sqrt{tN}}\sum_{ij}\xi_{ij}\omega(\sigma_i
\tau_{j}) =\alpha(1-\alpha) \beta^2 \left( 1 -\langle q_{12}p_{12} \rangle_t \right), \\
\mathcal{B} &=& \frac{-1}{N}\mathbb{E}\frac{\beta\sqrt{2(1-\alpha) \bar p}}{2\sqrt{1-t}}\sum_i
\eta_i\omega(\sigma_i) = -\frac{\alpha(1-\alpha)}{2}\beta^2\bar p\left( 1 -\langle q_{12} \rangle_t\right), \\
\mathcal{C} &=&
\frac{-1}{N}\mathbb{E}\frac{\beta\sqrt{2\alpha \bar q}}{2\sqrt{1-t}}\sum_{j}\tilde{\eta}_{j}\omega(\tau_{j})
= -\frac{\alpha(1-\alpha)}{2}\beta^2\bar q \left( 1 - \langle p_{12}\rangle_t \right).
\end{eqnarray}
So we can build the $t$-streaming of the interpolant $A_{N_1,N_2}(t)$ as
\begin{eqnarray}
\frac{d}{dt} A_{N_1,N_2}(t) &=& \mathcal{A}+\mathcal{B}+\mathcal{C}-\alpha(1-\alpha)\beta^2(1-\bar q)(1-\bar p)\nonumber\\
&=&\alpha(1-\alpha) \beta^2 \left( 1 -\langle q_{12}p_{12} \rangle_t \right)-\frac{\alpha(1-\alpha)}{2}\beta^2\bar p\left( 1 -\langle q_{12} \rangle_t\right)+\nonumber\\
&-&\frac{\alpha(1-\alpha)}{2}\beta^2\bar q \left( 1 - \langle p_{12}\rangle_t \right)-\alpha(1-\alpha)\beta^2(1-\bar q)(1-\bar p)\nonumber\\
&=&-\meanv{(\bar q-q_{12})(\bar p-p_{12})}_t.
\end{eqnarray}
Now we stress that the Lagrange multipliers $\bar q$ and $\bar p$
can be understood here as trial values for the order parameters.
According to this point of view, the replica symmetric condition,
\textit{i.e.} the request that the overlaps do not fluctuate, is
equivalent to impose $\meanv{(\bar q - q_{12})(\bar p -
p_{12})}_t=0$, bringing us to conclude that in RS regime
$A_{N_1,N_2}(t)$ is a steady function of $t$, and then
$$
A_{N_1,N_2}(t=1)=A_{N_1,N_2}(\beta)=A_{N_1,N_2}(t=0)= A_0(\alpha, \beta)+\alpha(1-\alpha)\beta^2(1-\bar q)(1-\bar p),
$$
that is, in the thermodynamic limit
\begin{eqnarray}
\bar A(\bar{p},\bar{q},\alpha, \beta)&=&\ln 2 + \alpha\mathbb{E}_g
\log \cosh
\big( g\beta\sqrt{2(1-\alpha)\bar{p}}\big) +\nonumber\\
&+& (1-\alpha) \mathbb{E}_g \log \cosh
\Large( g\beta\sqrt{2\alpha\bar{q}} \Large)+\alpha(1-\alpha)\beta^2(1-\bar q)(1-\bar p).\label{eq:A_trial}
\end{eqnarray}
Now we follow the same considerations exploited  in \cite{BGG1}.
Indeed, the  last expression holds barely for every possible
choice of the trial values $\bar q$ and $\bar p$ of the order
parameters $\meanv{q_{12}}$ and $\meanv{p_{12}}$. Our purpose is
then to fix the right value of $\bar q$ and $\bar p$, imposed by
the RS condition $\meanv{(\bar q - q_{12})(\bar p - p_{12})}=0$.
In primis we notice that the trial function $\bar A$, as a
function of the trial order parameters $\bar q$, $\bar p$, is
uniformly concave with respect to $\bar p$. In fact it  is easily
seen that
$$
\partial_{\bar p} \bar A(\bar q, \bar p, \alpha,\beta)=\alpha(1-\alpha)\beta^2\left( \bar q-\mathbb{E}_g\tanh^2 \left(g\beta\sqrt{2(1-\alpha)\bar p}\right) \right),
$$
and since $\mathbb{E}_g\tanh^2 \left(g\beta\sqrt{2(1-\alpha)\bar
p}\right)$ is  increasing in $\bar p$ we have the assertion.
Furthermore, for any fixed $\bar q$, the function $\bar A$ takes
its maximum value where the derivative vanishes, that defines
implicitly a special value for $\bar p(\bar q)$:
$$
\bar q=\mathbb{E}_g\tanh^2 \left(g\beta\sqrt{2(1-\alpha)\bar p(\bar q)}\right).
$$
Of course we have that $\bar p$ is an increasing function of $\bar
q$, with $\bar p(0)=0$.  Now we are concerned about $\bar A$ at a
fixed level set, \textit{i.e.} $A(\bar p(\bar q), \bar q )$, and
state that it is convex in $\bar q$. In fact it is easily seen
from the last formula that $\bar p(\bar q)/\bar q$ is an
increasing function of $\bar q$, but of course
$\mathbb{E}_g\tanh^2\left(g\beta\sqrt{2\alpha\bar q}\right)$ is
strictly decreasing, thus
$$
\partial_{\bar q} \bar A(\bar q, \bar p) =\alpha(1-\alpha)\beta^2\left( \bar p(\bar q)-\mathbb{E}_g\tanh^2 \left(g\beta\sqrt{2\alpha\bar q}\right) \right)
$$
is increasing in $\bar q$, that is equivalent to convexity in such
a variable. The last equation  specifies the right value of the
trial order parameter $\bar q$ in the RS approximation. Thus the
replica symmetric approximation of the pressure results uniquely
defined by the minimax principle:
\begin{theorem}
The replica symmetric free energy of the bipartite spin glass
model is uniquely defined by the following variational principle:
\begin{equation}
A^{RS}(\bar{p},\bar{q},\alpha, \beta)=\min_{\bar q}\max_{\bar p}
\bar A(\bar p, \bar q, \alpha, \beta),
\end{equation}
where
\begin{eqnarray}
\bar A(\bar p, \bar q, \alpha, \beta) &=& \ln 2 + \alpha\mathbb{E}_g
\log \cosh\big( g\beta\sqrt{2(1-\alpha)\bar{p}}\big) +\nonumber\\
&+& (1-\alpha) \mathbb{E}_g \log \cosh\Large(
g\beta\sqrt{2\alpha\bar{q}} \Large)+\alpha(1-\alpha)\beta^2(1-\bar
q)(1-\bar p)\label{eq:A_RS},
\end{eqnarray}
whose the saddle point is reached at the intersection of the
following two curves in the $(\alpha, \beta)$ plane
\begin{eqnarray}
\bar{q}(\alpha, \beta) &=& \mathbb{E}_g \tanh^2
(g\beta\sqrt{2(1-\alpha)\bar{p}(\alpha, \beta)})\label{eq:bar-q} \\
\bar{p}(\alpha, \beta) &=& \mathbb{E}_g \tanh^2
(g\beta\sqrt{2\alpha\bar{q}(\alpha, \beta)})\label{eq:bar-p}.
\end{eqnarray}
\end{theorem}
We can go further and put some constraints, imposing the two
curves to intersect away from $(\bar p=0, \bar q=0)$. Hence we
must have a precise relation among the slopes of such two curves
near the origin, \textit{i.e.} $\lim_{\bar q\to0}\bar p(\bar
q)/\bar q\geq\lim_{\bar p\to0} \bar p/\bar q(\bar p)$; but since
$\lim_{\bar q\to0}\bar p(\bar q)/\bar q=\alpha \beta^2$ and
$\lim_{\bar p\to0} \bar p/\bar q(\bar p)=1/\beta^2(1-\alpha)$, the
latter inequality simply leads us to conclude that only trivial
intersection point are possible for
$4\beta^4\alpha(1-\alpha)\leq1$. Therefore, we can resume all
these results in the following
\begin{proposition}
In the thermodynamic limit, it exists and it is unique the replica
symmetric pressure of the bipartite spin glass, given by
(\ref{eq:A_RS}). Furthermore the region of the $(\alpha,\beta)$
plane such that the (\ref{eq:bar-q}) and (\ref{eq:bar-p})  have
only the trivial solutions is characterized by
$4\beta^4\alpha(1-\alpha)\leq1$.
\end{proposition}
\begin{remark}
In fact for $4\beta^4\alpha(1-\alpha)\leq1$ the $\min\max$ is
obtained for $\bar q, \bar p=0$, and,  as it is easily seen, the
pressure (\ref{eq:A_RS}) reduces to the annealed one
(\ref{eq_A-ann}), that coincides with the true pressure of the
model in the thermodynamic limit in such a region.
\end{remark}

Furthermore, bearing in mind  (\ref{eq:A_RS}), together with
(\ref{eq:bar-q}) and (\ref{eq:bar-p}), it is a remarkable result
that the replica symmetric free energy of the bipartite model is
given by the convex combination of two monoparty spin glasses,
at different temperatures, exactly as happens in the ferromagnetic
case. This is clarified by the following
\begin{proposition}
Choosing $\beta'=\beta\sqrt{2\alpha}\sqrt{\frac{1-\bar q}{1-\bar p}}$ and $\beta''=\beta\sqrt{2(1-\alpha)}\sqrt{\frac{1-\bar q}{1-\bar p}}$, we have
\begin{equation}\label{eq:conv.RS}
A_{RS}(\alpha, \beta)=\alpha A^{SK}_{RS}(\beta')+(1-\alpha)A_{RS}^{SK}(\beta''),
\end{equation}
while, with the different  scaling of the inverse temperatures
$\beta'=\beta\sqrt{2\alpha}\sqrt{\frac{\bar q}{\bar p}}$ and
$\beta''=\beta\sqrt{2(1-\alpha)}\sqrt{\frac{\bar q}{\bar p}}$, we
have
\begin{equation}\label{eq:conv.RS-ann}
A_{RS}(\alpha, \beta)-A_A(\alpha, \beta)=\alpha \left(A^{SK}_{RS}(\beta')-A_A^{SK}(\beta')\right)+(1-\alpha)\left(A_{RS}^{SK}(\beta'')-A_A(\beta'')\right).
\end{equation}
\end{proposition}
\begin{proof}
The proof follows by a  straight calculation. If we take for the
two monoparty model the two different inverse temperatures
$\beta'=\beta\sqrt{2\alpha}a^2$ and
$\beta''=\beta\sqrt{2(1-\alpha)}a^2$, with $a$ a free parameter to
be determined at the end, we have
\begin{eqnarray}
A_{RS}(\alpha, \beta)&=&\ln 2 + \alpha\mathbb{E}_g \log \cosh
\big( g\beta\sqrt{2(1-\alpha)\bar{p}}\big) +\nonumber\\
&+& (1-\alpha) \mathbb{E}_g \log \cosh
\Large( g\beta\sqrt{2\alpha\bar{q}} \Large)+\alpha(1-\alpha)\beta^2(1-\bar q)(1-\bar p)\nonumber\\
&=&\ln 2 + \alpha\mathbb{E}_g \log \cosh
\big( g\beta''a^2\sqrt{\bar{p}}\big) +\nonumber\\
&+& (1-\alpha) \mathbb{E}_g \log \cosh
\Large( \frac{g\beta}{a^2}\sqrt{\bar{q}} \Large)+\sqrt{\alpha(1-\alpha)}\beta'\beta''(1-\bar q)(1-\bar p)\nonumber\\
&\leq&\alpha\left( \log2 +\mathbb{E}_g\log\cosh\left(g\beta''a^2\sqrt{\bar p}\right) +\frac{\beta''^2}{4}a^2(1-\bar p)^2\right)\nonumber\\
&+&(1-\alpha)\left( \log2 +\mathbb{E}_g\log\cosh\left(g\frac{\beta'}{a^2}\sqrt{\bar q}\right) +\frac{\beta'^2}{4a^2}(1-\bar q)^2\right).
\end{eqnarray}
Now it is easily seen  that in the last formula we get an equality
with $a^4=\frac{1-\bar q}{1-\bar p}$. Following exactly the same
path of the previous part of the work (dealing with ferromagnetic
models) we recover (\ref{eq:conv.RS-ann}) with the choice
$a^4=\frac{\bar q}{\bar p}$. We stress that this last value of $a$
is more meaningful, in the sense that in this case the two
monoparty models are trivially independent and separated, with
the order parameters given by usual self consistency relations for
the SK model:
\begin{eqnarray}
\bar q&=&\mathbb{E}_g \tanh^2(g\beta'\sqrt{\bar q}),\nonumber\\
\bar p&=&\mathbb{E}_g \tanh^2(g\beta''\sqrt{\bar p}),\nonumber
\end{eqnarray}
in perfect analogy with the ferromagnetic case.
\end{proof}
Lastly, as it is well known, a theory  with  no overlap
fluctuations allowed may not hold at low temperatures and we want
to report about its properties in the limit $\beta\to\infty$ to
check the stability of the replica symmetric ansatz. We will
concern about the ground state energy ${\hat e}_{RS}$ and its
associated entropy ${\hat s}_{RS}$, defined by
\begin{eqnarray} \label{ehat}
{\hat e}_{RS}(\alpha)&=&\lim_{\beta\to\infty} \partial_{\beta}{\bar A}_{RS}(\alpha,\beta)=\lim_{\beta\to\infty} {\bar A}_{RS}(\alpha,\beta)/{\beta},\\
\label{shat}
{\hat s}_{RS}(\alpha)&=&\lim_{\beta\to\infty}\big({\bar A}_{RS}(\alpha,\beta)-\beta \partial_{\beta}{\bar A}_{RS}(\alpha,\beta)\big).
\end{eqnarray}
First of all, from the self-consistency equations (\ref{eq:bar-q}, \ref{eq:bar-p}), through a long but straightforward calculation, we can compute the low temperature limit for the order parameters, ${\bar q}(\alpha,\beta),{\bar p}(\alpha, \beta)\to1$, together with the rates they approach to their limit value, $\beta(1-{\bar q}(\alpha,\beta))\to (\pi(1-\alpha))^{-1/2}$, $\beta(1-{\bar p}(\alpha,\beta))\to (\pi\alpha)^{-1/2}$. Then, bearing in mind the explicit form of the pressure of the model in the replica symmetric regime (\ref{eq:A_RS}), we derive the following expressions for the ground state energy and the entropy:
\begin{eqnarray} \label{ehat1}
{\hat e}_{RS}(\alpha)&=& \sqrt{\frac{\alpha(1-\alpha)}{\pi}},\\
\label{shat1}
{\hat s}_{RS}(\alpha)&=& -\frac{2}{\pi} \left(1- \sqrt{\alpha(1-\alpha)} \right).
\end{eqnarray}
Notice that the entropy ${\hat s}_{RS}(\alpha)$ is strictly less
than zero for every $\alpha\in(0,1)$, that is a typical feature of
the replica symmetric \textit{ansatz} for glassy systems.
Therefore, the true solution of the model must involve replica
symmetry breaking. Furthermore, it is a concave function of
$\alpha$, and assume its maximum value ${\hat s}_{RS}=-1/\pi$ in
$\alpha=1/2$, \textit{i.e.} the balanced bipartite, and its
minimum value ${\hat s}_{RS}=-2/\pi$ at the ending points
$\alpha=0,1$, when the size of one party is negligible in the
thermodynamic limit. Analogously,  for the ground state energy we
find in the perfectly balanced case ${\hat e}_{RS}(1/2)
=1/2\sqrt{\pi}$, and at the extrema of the definition interval of
$\alpha$ ${\hat e}_{RS}(0)={\hat e}_{RS}(1) =0$.

\subsection{Constraints}

While  the order parameters for simple models (as the bipartite
CW) are self-averaging, frustrated systems are expected to show
the replica symmetry breaking phenomenon \cite{MPV}\cite{broken},
which ultimately inhibits such a self-averaging properties for
$\langle q_{12} \rangle, \langle p_{12} \rangle$. As a consequence
a certain interest for the constraints to free overlap
fluctuations raised in the past
\cite{sum-rules}\cite{GG}\cite{AC}\cite{barra1} (and recently has
been deeply connected to ultrametricity
\cite{panchenko}\cite{arguin}) which motivate us to work out the
same constraints even in bipartite models.
\newline
To fulfil this task the first step is obtaining an explicit
expression for the internal energy density (which is
self-averaging \cite{contuccigiardina}).

\begin{theorem}\label{selfaverage}
The following expression for the internal energy density of the
bipartite spin glass model holds in the thermodynamic limit
\be\label{energy} \lim_{N \to \infty} \frac{1}{N} \langle
H_{N_1,N_2}(\sigma, \tau; \xi) \rangle = e(\alpha,\beta)=
2\alpha(1-\alpha)\beta^2 \Big(1- \langle q_{12}p_{12} \rangle
 \Big).
\ee
\end{theorem}
As the proof can be achieved by direct evaluation, we skip it and
turn to the constraints: Starting with the linear identities we
state the following
\begin{proposition}\label{ACcorollary}
In the thermodynamic limit, and $\beta$ almost-everywhere, the
following generalization of the linear overlap constraints holds
for the bipartite spin glass
 \be\label{ACprima}
 \langle q_{12}^2p_{12}^2 \rangle - 4 \langle
q_{12}p_{12}q_{23}p_{23} \rangle + 3 \langle
q_{12}p_{12}q_{34}p_{34}\rangle  =0.\ee
\end{proposition}
\begin{proof}
Let us address our task by looking at the $\beta$ streaming of the
internal energy density,  once expressed via $\langle q_{12}p_{12}
\rangle$; in a nutshell, physically, we obtain these constraints
by imposing that such a response can not diverge, neither in the
thermodynamic limit:
\begin{eqnarray}
\partial_{\beta} \langle q_{12}p_{12}\rangle &=&
\frac{1}{N_1 N_2}\sum_{i,j} \mathbb{E} \partial_{\beta} \omega^2
(\sigma_i\tau_{j}) = \frac{1}{N_1 N_2}\sum_{i,j} \mathbb{E} 2
\omega (\sigma_i \tau_{j})\partial_{\beta} \omega(\sigma_i
\tau_{j})
\\ &=&
\frac{2}{N_1 N_2}\sum_{i,j}\mathbb{E}\omega(\sigma_i
\tau_{j})\xi_{i\nu}\Big( \omega(\sigma_i \tau_{j} \sigma_j
\tau_{\nu}) -\omega(\sigma_i \tau_{j})\omega(\sigma_j
\tau_{\nu}) \Big),
\end{eqnarray}
now we use Wick theorem on $\xi$ and introducing the overlaps we
have
\begin{eqnarray}
\partial_{\beta} \langle q_{12}p_{12}\rangle &=& K\Big( \langle
p_{12}^2q_{12}^2 \rangle - \langle p_{12}q_{12}p_{13}q_{13}
\rangle \\ \nonumber &-&  \langle p_{12}q_{12}p_{13}q_{13} \rangle
+ \langle p_{12}q_{12}p_{34}q_{34} \rangle  + \langle \bar{p}
q_{12}p_{12} - \langle p_{12}q_{12}p_{13}q_{13} \rangle
\\ \nonumber
&-&  \langle p_{12}q_{12}p_{13}q_{13} \rangle +  \langle
p_{12}q_{12}p_{34}q_{34} \rangle - \langle \bar{p}q_{12}p_{12}
\rangle +  \langle p_{12}q_{12}p_{34}q_{34} \rangle.
\end{eqnarray}
The several cancelations leave the following remaining terms \be
\partial_{\beta} \langle q_{12}p_{12}\rangle = N_2 \Big( \langle q_{12}^2p_{12}^2 \rangle -4  \langle q_{12}p_{12}q_{23}p_{23} \rangle + 3 \langle q_{12}p_{12}q_{34}p_{34} \rangle \Big)
\ee and, again in the thermodynamic limit, the thesis is proved.
\end{proof}

\begin{proposition}\label{GG}
In the thermodynamic limit, and in $\beta$-average,  the following
generalization of the quadratic Ghirlanda-Guerra relations holds
for the bipartite spin glass
\begin{eqnarray}
\langle q_{12}p_{12}q_{23}p_{23} \rangle = \frac{1}{2}\langle
q_{12}^2p_{12}^2 \rangle + \frac{1}{2}\langle q_{12}p_{12}
\rangle^2, \\
\langle q_{12}p_{12}q_{34}p_{34} \rangle = \frac{1}{3}\langle
q_{12}^2 p_{12}^2 \rangle + \frac{2}{3}\langle q_{12}p_{12}
\rangle^2.
\end{eqnarray}
\end{proposition}
\begin{proof}
The idea is to impose, in the thermodynamic limit, the
self-averaging of the internal energy (i.e. $\langle
e^2(\alpha,\beta) \rangle - \langle e(\alpha,\beta) \rangle^2 =
0$); we obtain a rest that must be set to zero and gives the
quadratic control. Starting from
$$
 \mathbb{E}(e^2(\alpha,\beta)) = \frac{1}{(N_1+N_2)^3}
 \sum_{j,i}\sum_{\nu,j} \xi_{ij}\xi_{j\nu}\omega(\sigma_i
 \tau_{j}) \omega(\sigma_j \tau_{\nu}),
$$
 with a calculation perfectly analogous to the one
performed in the proof of Proposition \ref{ACcorollary} and
comparing with the former relations, we get the linear system
\begin{eqnarray}
0 &=& \langle q_{12}^2p_{12}^2 \rangle + 6 \langle
q_{12}p_{12}q_{34}p_{34} \rangle -6 \langle
q_{12}p_{12}q_{23}p_{23} \rangle - \langle q_{12}p_{12} \rangle^2,  \\
0 &=& \langle q_{12}^2p_{12}^2 \rangle -4  \langle
q_{12}p_{12}q_{23}p_{23} \rangle +3\langle
q_{12}p_{12}q_{34}p_{34} \rangle,
\end{eqnarray}
whose solutions give exactly the expressions reported in
Proposition \ref{GG}.
\end{proof}

\section{Conclusion}

In this paper we analyzed the equilibrium behavior of bipartite
spin systems (interacting both with ferromagnetic or with spin
glass couplings) trough statistical mechanics; these systems are
made of by two different subsets of spins (a priori of different
nature \cite{BGG1}\cite{io1}), for the sake of clearness each one
interacting with the other, but with no self-interactions. For the
former class trough several techniques, among which our mechanical
analogy of the interpolation method, early developed in
\cite{sum-rules} and successfully investigated in
\cite{barra0}\cite{io1}\cite{dibiasio},  we have seen that the
thermodynamic limit of the pressure does exist and it is unique
and we gave its explicit expression in a constructive way via a
minimax principle. Further, when introducing the Burger's equation
for the velocity field in our interpretation of the interpolating
scheme, our method automatically ''chooses'' the correct order
parameter, which turns out to be a linear combination of the
magnetizations of the two subsystems with different signs, so to
convert the minimax variational principle in a standard
extremization procedure.  Noticing that the same structure can be
recovered for many other models of greater interest, like
bipartite spin glasses, we went over and analyzed even the latter.
\newline
For these models we have studied both the annealing and the replica
symmetric approximation (the latter trough the double stochastic
stability technique recently developed in \cite{BGG1}) which
allowed us to give an explicit expression for the free energy  and
to discover and discuss the same minimax principle of the
bipartite ferromagnets.
\newline
Furthermore, we evaluated the replica symmetric observable in the
low temperature limit confirming the classical vision about the
need for a broken replica symmetry scheme: one step forward in
this sense, by studying the properties of the internal energy, we
derived all the classical constraints to the free overlap
fluctuations (suitable obtained for these systems) and we worked
out a picture of their criticality to conclude the investigation.
\newline
Future works on these subject should be addressed toward a
complete full replica symmetry broken picture and to a systematic
exploration of the multi-partite equilibria.

\section*{Acknowledgements}

Authors are grateful to MiUR trough the FIRB grant number
$RBFR08EKEV$ and to Sapienza University of Rome.
\newline
AB is partially funded by GNFM (Gruppo Nazionale per la Fisica
Matematica) which is also acknowledged.
\newline
FG is partially funded by INFN (Istituto Nazionale di Fisica
Nucleare) which is also acknowledged.
%\end{acknowledgements}



\end{document}